\documentclass[11pt]{article}

\usepackage{cite}
\usepackage{graphicx}
\usepackage{amsmath}
\usepackage{amsthm}
\usepackage{amsfonts}
\usepackage{amssymb}
\usepackage{fullpage}
\usepackage{geometry}
\usepackage{color}
\usepackage{latexsym}
\usepackage{algorithm}
\usepackage[noend]{algorithmic}
\usepackage{multirow}
\usepackage{arcs}
\usepackage{MnSymbol}

\newtheorem{theorem}{Theorem}[section]
\newtheorem{corollary}[theorem]{Corollary}
\newtheorem{lemma}[theorem]{Lemma}
\newtheorem{claim}[theorem]{Claim}
\newtheorem{observation}[theorem]{Observation}

\newtheorem{notn}[theorem]{Notation}

\DeclareMathOperator{\DT(P)}{DT(P)}

\newcommand{\old}[1]{{}}

\title{Bounded Degree Planar Geometric Spanners\thanks{Department of
Computer Science, Ben-Gurion University, Beer-Sheva 84105, Israel. 
P. C. and L. C. research is partially supported by Lynn and William Frankel
Center for Computer Science. }}

\author{
Paz Carmi \and Lilach Chaitman}

\begin{document}
\maketitle

\begin{abstract}
Given a set $P$ of $n$ points in the plane, we show how to compute in $O(n \log n)$ time a subgraph of their Delaunay triangulation that has maximum degree 7 and is a strong planar $t$-spanner of $P$ with $t =(1+ \sqrt{2})^2 *\delta$, where $\delta$ is the spanning ratio of the Delaunay triangulation. 
Furthermore, given a Delaunay triangulation, we show a distributed algorithm that computes the same bounded degree planar spanner in $O(n)$ time.
\end{abstract}

\section{Introduction}
Given a weighted graph $G=(V,E)$ and a real number $t \geq 1$, a $t$-\emph{spanner} of $G$ is a spanning subgraph $G^*$ with the property that for every edge $\{p,q\} \in G$, there exists a path between $p$ and $q$ in $G^*$ whose weight is no more than $t$ times the weight of the edge $\{p,q\}$. Thus, shortest-path distances in $G^*$ approximate shortest-path distances in the underlying graph $G$ and the parameter $t$ represents the approximation ratio. The smallest $t$, for which $G^*$ is a $t$-spanner of $G$, is known as the spanning ratio of the graph $G^*$. 

Spanners have been studied in many different settings. The various settings depend on the type of underlying graph $G$, on the way weights are assigned to edges in $G$, on the specific value of the spanning ratio $t$, and on the function used to measure the weight of a shortest path. We concentrate on the setting where the underlying graph is geometric. In this context, a geometric graph is a weighted graph whose vertex set is a set of points in $\Re^d$ and whose edge set consists of line segments connecting pairs of vertices. The edges are weighted by the Euclidean distance between their endpoints. 

There is a vast body of literature on different methods for constructing $t$-spanners with various properties in this geometric setting (see \cite{GiriSmid07} for a comprehensive survey of the area). Aside from trying to build a spanner that has a small spanning ratio, additional properties of the spanners are desirable, e.g., planarity and bounded degree.

In this paper we consider the following problem.
Given a set of points in the plane, the goal is to compute a bounded degree planar spanner of this set of points. Obtaining such a property enables us to perform routing algorithms that are known for plane graphs, e.g., \cite{BoseMorin04C} and  \cite{BoseMorin04}. Bose et al. \cite{BoseGS02} were the first to show the existence of a plane $t$-spanner (for some constant $t$) whose maximum vertex degree is bounded by a constant. Subsequently, Li and Wang \cite{LiW04} reduced the degree bound to 23. In \cite{BoseSX09}, Bose et al. improved the degree bound to 17. Currently, the degree bound stands at 14, as shown by Kanj and Perkovic \cite{PerkovicK08}. Here, we show how to reduce the degree even further.

More precisely, given a set $P$ of $n$ points in the plane and the Delaunay triangulation $\DT(P)$ of $P$, we show how to compute in $O(n \log n)$ time a subgraph of $\DT(P)$ that has maximum degree 7 and is a strong $t$-spanner of $P$ with $t =(1+\sqrt{2})^2 *\delta$, where $\delta$ is the spanning ratio of the Delaunay triangulation. This is a significant improvement in the degree bound from 14 to 7.
Furthermore, given a Delaunay triangulation and a clockwise order of the edges, we show a distributed algorithm that computes a planar spanner of degree 7 in $O(n)$ time. Notice that we do not require the edges to be sorted by their length. 

Another result shown in this paper is Corollary~\ref{cor:pathBound} which states the stretch factor for a special case of the Delaunay triangulation. We hope that this result will help shed some light on the real stretch factor of Delaunay triangulation.

\begin{notn}
Given a set of points $P$ in $\Re^2$, let $\DT(P)$ denote the Delaunay triangulation of $P$.
\end{notn}
\begin{notn}
Let $S_p= \{ q_0,\dots , q_k  \}$ be the set of neighbors of $p$ in $\DT(P)$, labeled in clockwise order.
\end{notn}
\begin{notn}
Let $S_{p,q_i,q_j} = \{q_k \in S_p: \ q_k$ is after $q_i$  and before  $q_j$ in the clockwise order  $\}$, for angle  $\angle(q_i p q_j) < \pi$.
\end{notn}
For simplicity of presentation, in the rest of the paper we assume that in $S_{p,q_i,q_j}$  the index of $i$ 
is smaller than the index of $j$ in the clockwise order.
\begin{notn}
Let wedge $W_{p,q_i,q_j}= \{\{p,q_k\}: q_k \in S_{p,q_i,q_j}  \}$ (see,~Figure~\ref{fig:wedge}).
\end{notn}
\begin{notn}
Let $P_{S_{p,q_1,q_k}}$ 
denote the path in $\DT(P)$ from $q_1$ to $q_k$ restricted to points in $S_{p,q_1,q_k}$.
\end{notn}
\begin{notn}
Let $\delta_{S_{p,q_1,q_k}}$ 
denote the length of the path $P_{S_{p,q_1,q_k}}$.
\end{notn}
\begin{notn}
Let $\{p,q_{min}\}$ be the shortest edges in $\DT(P)$ that is incident to $p$ ($q_{min} =  min_{q \in S_p} |\{p,q \}|$).
\end{notn}
\begin{notn}
For a graph $G = (V,E)$ and two points $p,q \in V$, let $\delta_G(p,q)$ denote the length of the shortest path between $p$ and $q$ in $G$.
\end{notn} 
\begin{notn}
For each $p \in P$, let $C_p=\{C_{p_1},C_{p_2}, \dots ,C_{p_8}  \}$ denote a set of 8 closed cones labeled in clockwise order, with apex $p$ and $\frac{\pi}{4}$ angle, such that the boundary of cones  $C_{p_1}$ and $C_{p_8}$ contains the edge $\{p,q_{min}\}$.
\end{notn}  
\begin{notn}
Let $D_{p,a,z}$ denote the disk having $p$, $a$ and $z$ on its boundary.
\end{notn}

\begin{figure}[htp]
    \centering
        \includegraphics[width=0.27\textwidth]{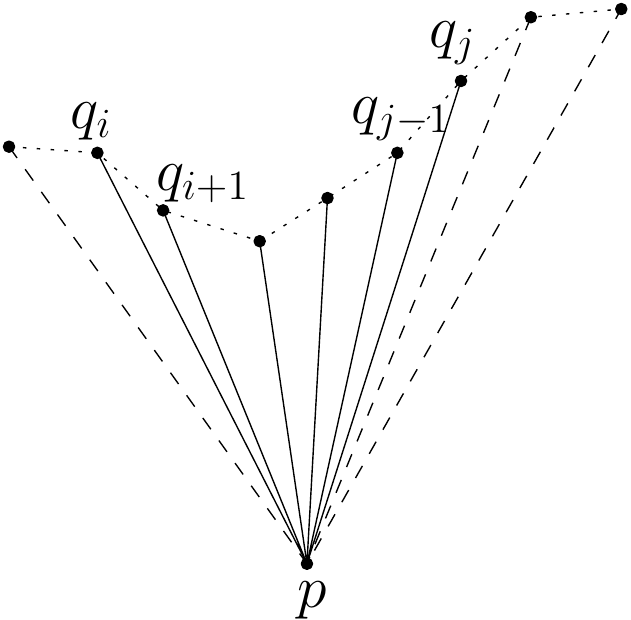}
    \caption{Wedge $W_{p,q_i,q_j}$ and its edges depicted in solid lines} 
    \label{fig:wedge}
\end{figure}

The rest of this paper is organized as follows.
In Section~\ref{sec:alg1} we describe an $O(n \log n)$ algorithm (Algorithm~\ref{alg:AlgMain}) that bounds the degree of each point $p \in P$ by $7$. In Section~\ref{sec:bou1} we show that this algorithm gives a sub-graph of the Delaunay triangulation with bounded degree $7$. Then, in Section~\ref{sec:SR1} we prove that the stretch factor of the sub-graph resulting from the algorithm is a small constant. Finally, in Section~\ref{sec:runTime} we show that given a Delaunay triangulation one can compute a distributed planar spanner of degree 7 in linear running time.  

\section{Algorithm for bounded degree 7} \label{sec:alg1}
In this section we describe an algorithm that computes a bounded degree planar spanner. The approach we take to build such a spanner is to start with the Delaunay triangulation and then prune its edges to achieve the degree bound of 7 while maintaining a constant spanning ratio. To achieve this we ensure that for every edge of the Delaunay triangulation where we do not add to our resulting spanner, there is a spanning path in the resulting subgraph that approximates this edge.

The algorithm consists of two main subroutines, BoundSpanner() and Wedge(). 
We start with the BoundSpanner() algorithm.
The first subroutine (Algorithm~\ref{alg:AlgMain}) is the main component of the algorithm. 
In the beginning of this subroutine, we compute the Delaunay triangulation ($\DT(P)$) of the set of points and sort them in nondecreasing length order.
We find for each point $p$ its nearest point $q_{min}$ in $P$; since $\{p,q_{min} \}$ is an edge of $\DT(P)$, we use the $\DT(P)$ edges to find this edge for each point. 
Then, for each point $p$ the orientation of the closed cones of $C_p$ is defined according to this shortest
edge $\{p,q_{min}\}$, such that this edge is shared by cones $C_{p_1}$ and $C_{p_8}$.  
Next, an edge $\{p,q\}$ is added to $G$ if both $p$ and $q$ \emph{agree} on it. A point $p$ \emph{agrees} on an edge $\{p,q\}$ if every cone $C_{p_i} \in C_p$ containing $\{p,q\}$ is empty (i.e., $C_{p_i} \cap E = \emptyset$ for every cone $C_{p_i} \in C_p $ containing the edge $\{p,q\}$).
Notice that an edge could be contained in at most two closed cones, and exactly, two cones if this edge is on the common boundary of two cones. After adding an edge $\{p,q\}$, we call the second subroutine (\emph{Wedge()}) twice, once for point $p$ and then for point $q$.   
In the second subroutine (Algorithm~\ref{alg:wedge}) we add more edges to $E$ that do not affect the degree of the spanner;  however, they help us to bound the stretch factor of the resulting graph.

\begin{algorithm}[htp]
\caption{BoundSpanner($P$)}\label{alg:AlgMain}
\begin{algorithmic}[1]
\REQUIRE A set $P$ of points in the plane
\ENSURE A planar $t$-spanner $G = (P,E)$ with maximum degree 7
\vspace{0.2cm}
\STATE Compute $\DT(P)\rightarrow (P,E_{DT})$
\STATE Let $L$ be a sorted list of the edges of $DT$ by nondecreasing length
\STATE $E \leftarrow \emptyset$ 
\STATE $E^* \leftarrow \emptyset$ 
\STATE  Initialize $C_p$ for each $p \in P$ \hspace{3.5cm} /* with respect to edge $\{p,q_{min}\}$ */  

\vspace{0.1cm}
\FOR {each edge $\{p,q \}\in L$    (* in the sorted order *)} \label{loop1}   
	   \IF {($\forall C_{p_i}$ contains $\{p,q  \}$, $C_{p_i}\cap E = \emptyset$)  
	                and ($\forall C_{q_j}$ contains $\{p,q  \}$, $C_{q_j}\cap E = \emptyset$)\\
	     \qquad  /* Note: every edge can be contained in at most two adjacent closed cones */\\
	   }
				\STATE $E \leftarrow E\cup \{ \{p,q  \} \}$
		    \STATE Wedge($p,q$)        \hspace{3.1cm}  /*   calling  subroutine Wedge() to check if some 
        \STATE Wedge($q,p$)     \hspace{3.7cm}  edges ($E^*$) are needed to be added to $E$ \ \ */
    \ENDIF
\ENDFOR
\STATE $E \leftarrow E\cup E^*$ 

\end{algorithmic}
\end{algorithm}

\begin{algorithm}[htp]
\caption{Wedge($p,q_i$)}\label{alg:wedge}
\begin{algorithmic}[1]
\REQUIRE Two points $p$ and $q_i$ such that the edge $\{p,q_i \} \in \DT(P)$.
\ENSURE A set of edges $E^*$ to be added to the spanner $G = (P,E)$ 
\vspace{0.2cm}
\FOR {every $C_{p_z}$ contains $\{p,q_i \}$}
	\STATE Let $\{p,q_j \}$ and $\{p,q_k \}$ be the first and the last edges in cone $C_{p_z}$ (clockwise).
  \STATE $E^* \leftarrow E^* \cup \{ \{ q_m,q_{m+1} \} \}$ for each $j< m < i-1$\label{step:wedAdd1}
	\STATE $E^* \leftarrow E^* \cup \{ \{ q_m,q_{m+1} \} \}$ for each $i< m < k-1$\label{step:wedAdd2}
	
	\IF {(edge $\{p,q_{i+1} \} \in C_{p_z}$) and ($q_{i+1} \neq q_k$) and   (angle $\angle(pq_iq_{i+1}) > \pi/2$) }
          \STATE $E^* \leftarrow E^* \cup \{ \{ q_i,q_{i+1} \} \}$\label{step:wedAdd3}
	\ENDIF 
	\IF {(edge $\{p,q_{i-1} \} \in C_{p_z}$) and ($q_{i-1} \neq q_j$) and  (angle $\angle(pq_iq_{i-1}) > \pi/2$) }
           \STATE $E^* \leftarrow E^* \cup \{ \{ q_i,q_{i-1} \} \}$\label{step:wedAdd4}
	\ENDIF 
\ENDFOR

\end{algorithmic}
\end{algorithm}

\textbf{Remark:} Note that the output $t$-spanner $G$ of $P$ obtained by Algorithm~\ref{alg:AlgMain} is a subgraph of $\DT(P)$, therefore, it is planar and has a linear number of edges. 

\section{Bounded degree}\label{sec:bou1}
In this section we show that the degree of a point in the resulting graph is at most 7. We start by making two basic observations, then conclude with Lemma~\ref{lemma:lemma0}.
 
\begin{observation}\label{obs:obs1}
An edge $\{p,q\}\in C_{p_i}$ is added to $E$ during the loop in Algorithm~\ref{alg:AlgMain} step~\ref{loop1} only if $C_{p_i} \cap E = \emptyset$.
\end{observation}

\begin{observation}\label{obs:commonEdge}
The first edge incident to a point $p$ added to $E$ during Algorithm~\ref{alg:AlgMain} is $\{p,q_{min}\}$, thus, 
it belongs to two closed cones  $C_{p_1}$ and $C_{p_8}$ in $C_p$.
\end{observation} 
\begin{proof}
Consider the time when edge $\{p,q_{min}\}$ is considered to be added to $E$. Then, since it is the shortest edge incident to  $p$ in $\DT(P)$ it implies that all the cones in $C_p$ are empty; thus, $p$ ``agrees'' on adding edge $\{p,q_{min}\}$ to $E$. Let $C_{q_j}$ be the cone in $C_q$ that contains the edge $\{p,q_{min}\}$. Since the disk centered at $p$ with radius  $|p,q_{min}|$ is empty of points, it implies that cone $C_{q_j}$ that contains $\{p,q_{min}\}$ is empty. Thus, $q$ ``agrees'' on adding edge $\{p,q_{min}\}$ as well, and the edge is added to $E$.
\end{proof}

\begin{lemma}\label{lemma:lemma0}
The degree of spanner $G$ constructed by the above algorithm is bounded by 7.
\end{lemma}
\begin{proof}
Eight closed cones $C_{p}$ are defined for each point $p \in P$ during Algorithm~\ref{alg:AlgMain}.
By Observation~\ref{obs:commonEdge}, there are two cones $C_{p_1}$ and $C_{p_8}$ in $C_p$ sharing a common edge. 
Consider the edges $E^1_p$ incident to $p$ that are added to $E$ during Algorithm~\ref{alg:AlgMain} (not including the 
edges added during Algorithm~\ref{alg:wedge}). Then each edge $e \in E^1_p$ is added to $E$ only if the cone in $C_p$
containing $e$ is empty. Moreover, the first edge in $E^1_p$ added to $E$ shares two cones, thus $|E^1_p| \leq 7$.
Next, we show that the edges added during Algorithm~\ref{alg:wedge} can be charged uniquely to empty cones, thus
not increasing the degree bound of 7.
Let $\{ p, q\}$ be an edge added to $E$ during Algorithm~\ref{alg:wedge}; thus, there exists a point $z$
such that the edge $\{ p, q\}$ has been added to $E$ during the call \emph{Wedge(z,r)}. 
Moreover, this edge has been added in steps~\ref{step:wedAdd1}~,~\ref{step:wedAdd2} or during steps~\ref{step:wedAdd3}~,~\ref{step:wedAdd4}.
\begin{itemize}

	\item \textbf{Case 1:} The edge has been added during step~\ref{step:wedAdd1} or \ref{step:wedAdd2}. \\
        Let $\{p,s\}$ be the consecutive edge to $\{ p, z\}$, such that $q \neq s$. Since the edge $\{ p, q\}$
        has been added to $E$ during the call \emph{Wedge(z,r)}	it follows that the edge $\{z,s\}$ is in the same 
        cone (of $C_z$) as edges $\{z,p\}$ and $\{z,q\}$. Thus, the angle $\angle(szq) \leq \pi/4$
        and by the empty cycle property of Delaunay triangulation angle $\angle(qps) \geq 3\pi/4$.
        Therefore, there are at least two empty cones of $C_p$ located between $\{p,q\}$ and $\{p,s\}$.
        One of them is charged for the edge $\{p,q\}$ and the second is left to be charged for the edge $\{p,s\}$ 
        if needed.         
  \item \textbf{Case 2:} The edge has been added during step~\ref{step:wedAdd3} or \ref{step:wedAdd4}. \\
           In this case $r = p$. We know that the angle  $\angle(qpz) \geq \pi/2$, thus there is at least one empty cone $c'$ of $C_p$ located between $\{p,q\}$ and $\{p,z\}$. Therefore, this empty cone $c'$ is charged for the edge $\{p,q\}$.

 \end{itemize}

Therefore, the degree of every point $p \in P$ is bounded by 7.
\end{proof}

\section{Spanning ratio}\label{sec:SR1}
In this section we show that the spanning ratio of the resulting sub-graph is bounded. The empty circle property of Delaunay triangulations  allows us to make two basic but crucial observations.
\begin{observation}\label{obs:in-D}
From the empty cycle property of Delaunay triangulation it follows that each $x \in S_{s,r,p}$ is inside $D_{s,p,r}$.
\end{observation}
\begin{observation}\label{obs:pi-x}
For $q_j,q_i,q_k \in S_{p,q_j,q_k}$ such that $q_i$ is between  $q_j$ and $q_k$ in the clockwise order, the angle $\angle{q_j q_i q_k} \geq \pi-\angle{q_j p q_k}$.
\end{observation}
\begin{proof}
Due to the empty cycle property of Delaunay triangulation, the point $q_i$ lies inside the disk $D_{p,q_j,q_k}$ having $p,q_j,q_k$ on its boundary (Observation~\ref{obs:in-D}). The angle $\angle{q_j q_i q_k}$ is minimized when $q_i$ is on the boundary of $D_{p,q_j,q_k}$. In that case $\angle{q_j q_i q_k} = \pi-\angle{q_j p q_k}$ since the two angles lie on the same chord $(q_j,q_k)$. Therefore, $\angle{q_j q_i q_k} \geq \pi-\angle{q_j p q_k}$.
\end{proof}

\begin{observation} \label{obs:x/sinx}
Let $D_{p,a,z}$ be a disk having $p$, $a$, and $z$ on its boundary and let $\beta$ denote the angle $\angle(pza)$. Then, $\frac{\beta}{\sin(\beta)}|\{p,a\}|$ is the length of the arc from $p$ to $a$ on the boundary of $D_{p,a,z}$ ($\wideparen{pa}$ ).
\end{observation}
\begin{proof}
Let $o$ be the center of $D_{p,a,z}$ and let $r$ be the length of its radius, thus, angle $\angle(poa)= 2 \beta$. By the law of sines, $\frac{|pa|}{\sin(\beta)} = 2r$. Therefore, the length of the arc $\wideparen{pa}$ is 
$$2\pi r/ \frac{2\pi}{2\beta} = 2\beta r = \frac{2\beta |pa|}{2\sin(\beta)} = \frac{\beta }{\sin(\beta)}|pa|.$$
\end{proof}
\begin{lemma}\label{lemm:pathBound1}
Consider a wedge $W_{s,r,p}$ in $\DT(P)$ and assume that $\{s,r\},\{s,p\}$ are the shortest edges in $W_{s,r,p}$ incident to $s$ (i.e., $|sr|,|sp| \leq |sx|$ for all $x \in S_{s,r,p} \backslash \{r,p\} $). 
Then, $\delta_{S_{s,r,p}}(r,p) \leq |rp|\frac{\alpha}{\sin(\alpha)}$, where $\alpha = \angle(rsp)$.
\end{lemma}
\begin{proof}
We prove the claim by induction on the rank of the angle $\alpha$, i.e., the place of $\alpha$ in a nondecreasing order of the angles in $\DT(P)$.

\noindent
{\bf Base case:} 
Angle $\alpha$ is the smallest angle in $\DT(P)$, thus, $\{r,p\} \in \DT(P)$ and clearly $\delta_{S_{s,r,p}}(r,p) \leq |rp|$.\\
{\bf The induction hypothesis:} 
Assume that for every $\alpha' < \alpha$ the claim holds.\\
{\bf The inductive step:}
If $W_{s,r,p}\backslash \{\{s,r\},\{s,p\}\} = \emptyset$, then $\{r,p\} \in \DT(P)$ and we are done. Otherwise, let $\{s,a\}$ be the shortest edge in $W_{s,r,p}\backslash \{\{s,r\},\{s,p\}\}$, i.e., $\{s,a\} = \min_{x \in S_{s,p,r} \backslash \{p,r \}  }\{ |sx| \}$ (see, Figure~\ref{fig:pathBound1}). 

Put $\alpha_1 = \angle(rsa)$ and  $\alpha_2 = \angle(psa)$. Since $\alpha_1, \alpha_2 < \alpha$, by the induction hypothesis, $\delta_{S_{s,r,a}}(r,a) \leq |ra|\frac{\alpha_1}{\sin(\alpha_1)}$, and $\delta_{S_{s,a,p}}(a,p) \leq |ap|\frac{\alpha_2}{\sin(\alpha_2)}$. Notice that $S_{s,r,a} \subseteq S_{s,r,p}$ and $S_{s,a,p} \subseteq S_{s,r,p}$. Thus, 
\begin{eqnarray*}
\delta_{S_{s,r,p}}(r,p) & = & \delta_{S_{s,r,a}}(r,a) + \delta_{S_{s,a,p}}(a,p)  \\
       &\leq &
|ra|\frac{\alpha_1}{\sin(\alpha_1)} + |ap|\frac{\alpha_2}{\sin(\alpha_2)}.
\end{eqnarray*}

Note that by observation~\ref{obs:in-D} $a$ is located inside $D_{s,p,r}$. Let $a'$ be the intersection point of $D_{s,p,r}$ and the extension of $\{s,a\}$. Since $|sp|\leq |sa|$, $\angle(sap) \leq \frac{\pi}{2}$, therefore, $\angle(a'ap)\geq \frac{\pi}{2}$ and $|a'p| \geq |ap|$. Symmetrically, we get $|ra'| \geq |ra|$ and therefore, 
$$\delta_{S_{s,r,p}}(r,p) \leq |ra'|\frac{\alpha_1}{\sin(\alpha_1)} + |a'p|\frac{\alpha_2}{\sin(\alpha_2)}.$$
According to Observation~\ref{obs:x/sinx}, $|ra'|\frac{\alpha_1}{\sin(\alpha_1)}$ and $|a'p|\frac{\alpha_2}{\sin(\alpha_2)}$ are the lengths of the arcs from $r$ to $a'$ ($\wideparen{ra'}$) and from $a'$ to $p$ ($\wideparen{a'p}$) on the boundary of $D_{s,p,r}$, respectively. Moreover, $|rp|\frac{\alpha}{\sin(\alpha)}$ is the length of the arc from $r$ to $p$ on the boundary of $D_{s,p,r}$, which is the sum of $\wideparen{ra'}$ and $\wideparen{a'p}$. Therefore, $$\delta_{S_{s,r,p}}(r,p) \leq |ra'|\frac{\alpha_1}{\sin(\alpha_1)} + |a'p|\frac{\alpha_2}{\sin(\alpha_2)} = |rp|\frac{\alpha}{\sin(\alpha)}.$$
\end{proof}

\begin{figure}[htp]
    \centering
        \includegraphics[width=0.27\textwidth]{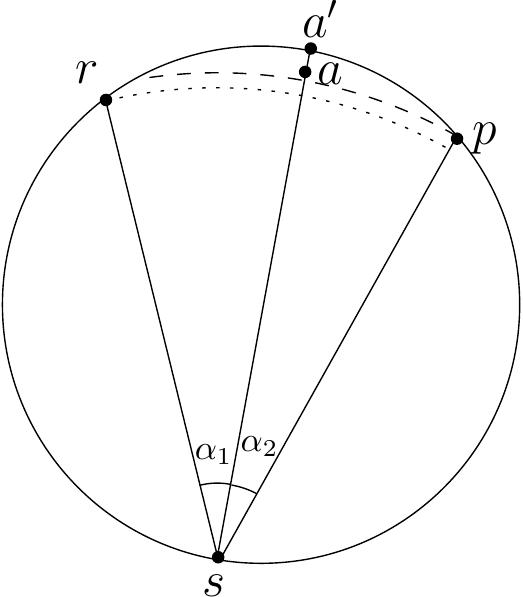}
    \caption{Illustrating the proof of Lemma~\ref{lemm:pathBound1} }
    \label{fig:pathBound1}
\end{figure}

\begin{claim} \label{claim:rightTri}
Let $\triangle(pqr)$ be a right-angled triangle with hypotenuse $(p,r)$. Then, $$\frac{\pi}{2\sqrt{2}}(|pq| + |qr|) \leq \frac{\pi}{2}|pr|.$$
\end{claim}
\begin{proof}
Let $\beta$ denote the angle $\angle(prq)$. By the law of sines, $|pq| = |pr|\sin(\beta)$ and $|qr| = |pr|\cos(\beta)$. Therefore,
\begin{eqnarray*}
\frac{\pi}{2\sqrt{2}}(|pq| + |qr|) & = &  \frac{\pi}{2\sqrt{2}}|pr|(\sin(\beta) + \cos(\beta)) \\
 & \leq^{(*)} & 
\frac{\pi}{2\sqrt{2}}|pr|(\frac{1}{\sqrt{2}} + \frac{1}{\sqrt{2}}) \\
 & = & \frac{\pi}{2\sqrt{2}}(\frac{2}{\sqrt{2}})|pr|  \\
 & = &\frac{\pi}{2}|pr|. 
\end{eqnarray*}
(*) The sum $\sin(\beta) + \cos(\beta)$ is maximized when $\beta = \frac{\pi}{4}$. 
\end{proof}

\begin{lemma}
\label{lemma:pathBound}
Let $\DT(P)$ be the Delaunay triangulation of the set of points $P$ and let $W_{s,r,p}$ be a wedge in $\DT(P)$, such that $\{s,r\}$ is the shortest edge in $W_{s,r,p}$ ($|sr| \leq |sx| \ \forall x \in S_{s,r,p}$). Let $r'$ be the projection of $r$ on $\{s,p\}$. Then, $$\delta_{S_{s,r,p}}(r,p) \leq \frac{\pi}{2\sqrt{2}}(|pr'| + |r'r|),$$  where $\alpha = \angle(rsp) \leq \frac{\pi}{4}$ (see,~Figure~\ref{fig:pathBound}).
\end{lemma}

\begin{proof}
We prove the lemma by induction on the rank of the angle $\alpha$.

\noindent 
{\bf Base case:} Angle $\alpha$ is the smallest angle in $\DT(P)$; therefore, $\{r,p\} \in \DT(P)$ and clearly $\delta_{S_{s,r,p}}(r,p) \leq |rp|$.

\noindent
{\bf The induction hypothesis:} 
Assume the claim holds for every angle $\alpha' < \alpha$.

\noindent
{\bf The inductive step:}
If $S_{s,r,p} \backslash \{r,p \} = \emptyset$, then $\{r,p\} \in \DT(P)$ and we are done. Otherwise, recall that from the empty cycle property of Delaunay triangulation it follows that each $x \in S_{s,r,p}$ is inside $D_{s,p,r}$. Let $a \in S_{s,r,p}$ be a point such that for every $x \in S_{s,r,p}\backslash \{r\}$, $|sa|\leq |sx|$. If $|sa| \geq |sp|$ by Lemma~\ref{lemm:pathBound1} 
$$\delta_{S_{s,r,p}}(r,p) \leq \frac{\alpha}{\sin(\alpha)}|rp| \leq \frac{\pi}{2\sqrt{2}}|rp| \leq \frac{\pi}{2\sqrt{2}}(|pr'|+|r'r|)$$ and we are done. 
Otherwise, ($|sa| < |sp|$), let $a'$ be the projection of $a$ on $\{s,p\}$. Denote $\alpha_1 = \angle(asp)$ and $\alpha_2 = \angle(asr)$. Since $\alpha_1 < \alpha$, we can apply the induction hypothesis and get 
$$ \hspace{4.2cm} \delta_{S_{s,r,a}}(a,p) \leq \frac{\pi}{2\sqrt{2}}(|pa'| + |a'a|). \hspace{4.9cm} (1) $$ Moreover, by Lemma~\ref{lemm:pathBound1}   
\begin{eqnarray*}
 \hspace{3.8cm} \delta_{S_{s,r,a}}(r,a) & \leq & \frac{\alpha_2}{\sin(\alpha_2)}|ra| \hspace{5.3cm} (2) \\  
         & \leq &  \frac{\alpha}{\sin(\alpha)}|ra| \\
         & \leq & \frac{\pi}{2\sqrt{2}}|ra|  \qquad \qquad (since, \  \alpha_2 \leq \alpha \leq \frac{\pi}{4}).
\end{eqnarray*}
Therefore, 
\begin{eqnarray*}
 \delta_{S_{s,r,p}}(r,p)  & \leq^{\ \  } & \delta_{S_{s,r,a}}(r,a) + \delta_{S_{s,a,p}}(a,p)  \\
       & \leq^{(2)} & \frac{\pi}{2\sqrt{2}}|ra| + \delta_{S_{s,a,p}}(a,p)  \\
       & \leq^{(1)} & \frac{\pi}{2\sqrt{2}}|ra| + \frac{\pi}{2\sqrt{2}}(|pa'| + |a'a|) \\ 
          & =^{\ \ } & \frac{\pi}{2\sqrt{2}}(|ra| + |pa'| + |a'a|) \\
          & \leq^{(*)} & \frac{\pi}{2\sqrt{2}}(|pr'| + |rr'|).
\end{eqnarray*}
The last inequality (*) is obtained by the following. Let $b$ denote the projection of $a$ on $(r',r)$. Thus, $|ba| = |a'r'|$ and  $|br'| = |a'a|$. Therefore,
\begin{eqnarray*}
|ra| + |pa'| + |a'a|  & \leq^{(**) \ \ } & |rb| + |ba| + |pa'| + |aa'|  \\
                      & =^{\qquad} & |rb| + |a'r'| + |pa'| + |br'| \\ 
                      & =^{\qquad} & |pr'| + |rr'|. 
\end{eqnarray*}
Inequality (**) follows by triangle inequality, $|ra| \leq |rb| + |ba|$.
\end{proof}

\begin{figure}[htp]
    \centering
        \includegraphics[width=0.27\textwidth]{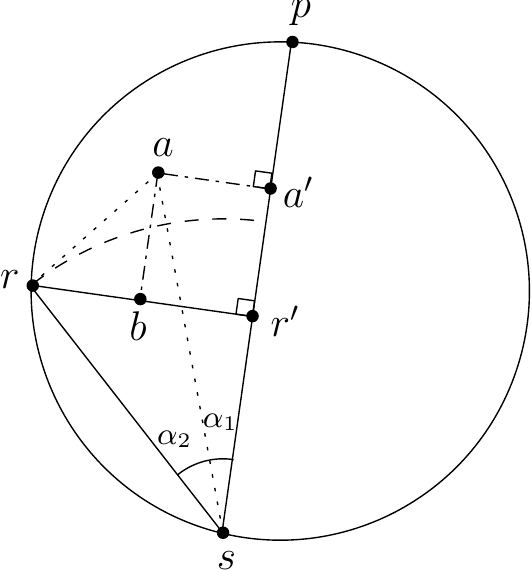}
    \caption{Illustrating the proof of Lemma~\ref{lemma:pathBound} }
    \label{fig:pathBound}
\end{figure}

Claim~\ref{claim:rightTri} and Lemma~\ref{lemma:pathBound} leads to the following corollary.
\begin{corollary}\label{cor:pathBound}
Let $\DT(P)$ be the Delaunay triangulation of set of points $P$ and let $s,r,p$ be points in $P$. If $\{s,r\}$ is the shortest edge in $W_{s,r,p}$, then $$\delta_{\DT(P)}(r,p) \leq \frac{\pi}{2}|rp|,$$ where $\alpha = \angle(rsp) \leq \frac{\pi}{4}$.
\end{corollary}

\begin{lemma}\label{lemma:shortest}
Let $\{s,r\}, \{s,p\}$ be two edges in $ C_{s_i} \cap \DT(P)$ such that $\{s,r\}$ has been chosen by Algorithm~\ref{alg:AlgMain} to be added to $E$. Then for every $\{s,x\} \in W_{s,r,p}, |sx| \geq min\{|sr|,|sp|\}$.
\end{lemma}

\begin{proof}
Assume, to the contrary, there is a point $\{s,x\} \in W_{s,r,p}$ such that $|sx| < min\{|sr|,|sp|\}$. Let $\{s,w\}$ be the shortest edge among all the edges in $W_{s,r,p}$. Since $|sw|<|sr|$, when $\{s,w\}$ was examined by Algorithm~\ref{alg:AlgMain}, the cone $C_{s_i}$ was empty from edges in $E$. Therefore, the only possible reason that could cause $\{s,w\}$ not to be added to $E$ is that $E$ already contained an edge in the cone with apex $w$ that $\{s,w\}$ belongs to (w.l.o.g to $C_{w_j}$). Let $\{w,t\}$ be an adjacent edge to  $\{w,s\}$ in $C_{w_j}$. Necessarily $\{t,s\} \in \DT(P)$, and it is also in $W_{s,r,p}$. However, 
$$\angle(wts) = \pi - \angle(wst) - \angle(swt) \geq \pi - 2\cdot\frac{\pi}{4} = \frac{\pi}{2},$$ we get $|st|<|sw|$ in contradiction to the assumption that $\{s,w\}$ is the shortest edge among all the edges in $W_{s,r,p}$.
\end{proof}

\begin{claim}\label{claim:triangle}
Let $\triangle(rqp)$ be a triangle with $\angle(rqp) \geq \frac{3\pi}{4}$, then $k|qp|+ d|rq|\leq k|rp|$ for $k \geq \sqrt{2}d$.
\end{claim}

\begin{proof}
 Let $q'$ be the point on $\{r,p\}$, such that $|qp|=|q'p|$ (see,~Figure~\ref{fig:triangle}). Since $\angle(rqp)\geq \frac{3\pi}{4}$, then $|rp|>|qp|$, therefore such a point exists.
 
 Since $|rp|=|rq'|+|q'p|=|rq'|+|qp|$, all there is left to prove is $$ k|qp|+ d|rq|\leq k(|rq'|+|qp|),$$
 which is equivalent to $$ d|rq|\leq k|rq'|.$$ 
 Denote $\angle(pqq')=\angle(pq'q)=\beta$, $\angle(qq'r)=\delta$, and $\angle(q'qr)=\gamma$.  
 Notice that $\beta < \frac{\pi}{2}$, therefore, $\frac{\pi}{2} < \delta <\pi$, and 
 $\frac{\pi}{4} < \gamma <\frac{\pi}{2}$.
 
 By the law of sines, 
 $$\frac{|rq|}{|rq'|}=\frac{\sin(\delta)}{\sin(\gamma)|}\leq\frac{1}{\sin(\frac{\pi}{4})}=\sqrt{2}$$ 
 Thus, $|rq| \leq |rq'|\sqrt{2}$, which finishes the proof.
\end{proof}

\begin{figure}[htp]
    \centering
        \includegraphics[width=0.45\textwidth]{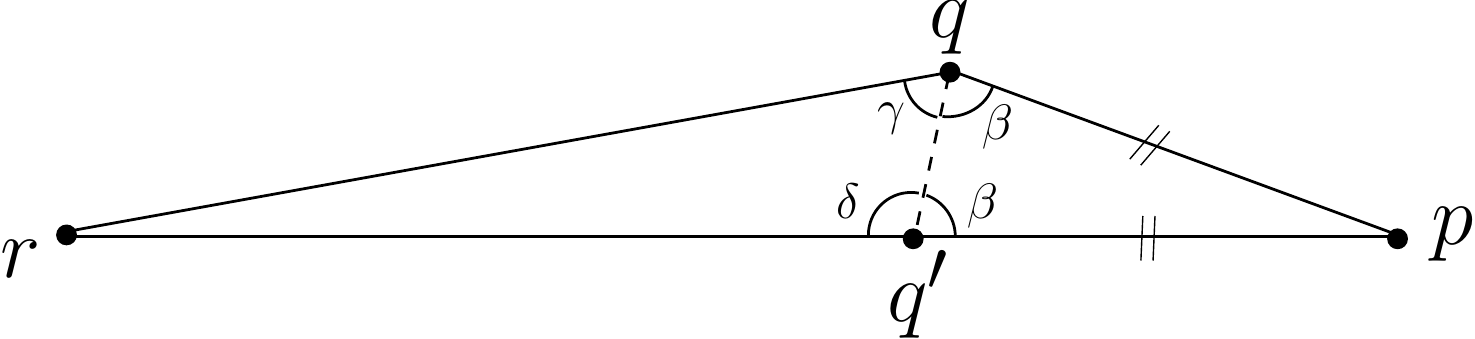}
    \caption{Illustrating the proof of Claim~\ref{claim:triangle}}
    \label{fig:triangle}
\end{figure}

\begin{claim}
\label{cl:rr'AndSp} Let $\{s,r\}$ and $\{s,p\}$ be two edges in $\DT(P)$, such that $|sr| \leq |sp|$ and the angle between $\{s,r\}$ and $\{s,p\}$ ($\angle(rsp)$) is less than $\pi/4$. Then, $$|sr| + K(|rr'|+|r'p|) \leq K|sp|$$ for $K \geq \frac{1}{1- 2 \sin(\frac{\pi}{8})}$, where $r'$ is a point on $\{s,p\}$ such that $|sr'| = |sr|$.
 \end{claim}
 \begin{proof}
  Put $\angle(r'sr)=\alpha$ and $\angle(srr')=\beta= \frac{\pi-\alpha}{2}$ (see,~Figure~\ref{fig:rr'AndSp}); then by the law of sines, $|rr'|= \frac{\sin(\alpha)}{\sin(\beta)}|sr'|$.	Therefore,
 	\begin{eqnarray*}
 	|sr| + K(|rr'|+|r'p|)&\leq& |sr| + K(\frac{\sin(\alpha)}{\sin(\beta)}|sr'|+ |r'p|) \\
 	     &=&  |sr| + K(\frac{\sin(\alpha)}{\sin(\beta)}|sr'|+ |sp| - |sr'|) \\
       &=&  |sr| + K((\frac{\sin(\alpha)}{\sin(\beta)} - 1)|sr'|+ |sp|) \\
       &=& |sr| + K((\frac{\sin(\alpha)}{\sin(\frac{\pi-\alpha}{2})}-1)|sr'|+ |sp|) \\
 	     &=& |sr| + K((\frac{\sin(\alpha)}{\cos(\alpha/2)}-1)|sr'|+ |sp|) \\
 	     &=& |sr| + K(2\sin(\alpha/2)-1)|sr'|+ |sp|) \\
 	     &=& |sr|(1 + K(2\sin(\alpha/2)-1)) + K|sp|) \\
 	     &\leq & |sr|(1 + K(2\sin(\pi/8)-1)) + K|sp|) \\
 	     &\leq^{(*)}& K|sp|.
 	\end{eqnarray*}     
 	The last inequality (*) follows from the fact that $(1 + K(2\sin(\alpha/2)-1))$ is less than zero for $K \geq \frac{1}{1- 2 \sin(\frac{\pi}{8})}$.
\end{proof}		

\begin{figure}[htp]
    \centering
        \includegraphics[width=0.6\textwidth]{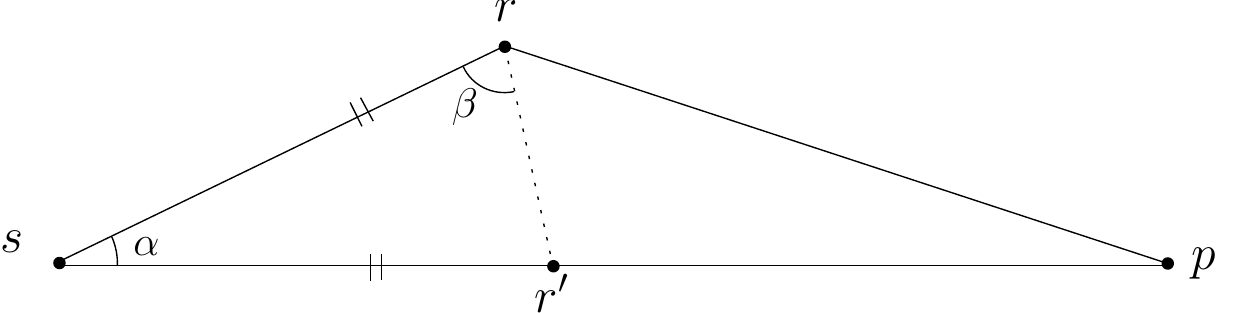}
    \caption{Illustrating the proof of Claim~\ref{cl:rr'AndSp}}
    \label{fig:rr'AndSp}
\end{figure}

\begin{lemma}\label{lemma:lemma5}
The stretch factor of the resulting $t$-spanner of Algorithm~\ref{alg:AlgMain} is $(1+\sqrt{2})^{2}\cdot\delta$, where $\delta$ is the stretch factor of Delaunay triangulation, i.e., $t =(1+\sqrt{2})^{2}\cdot\delta$. 
\end{lemma}
\begin{proof}
Let $G=(P,E)$ be the output graph of Algorithm~\ref{alg:AlgMain}.
To prove the Lemma we show that for every edge $\{s,p\} \in \DT(P)$, $\delta_G(s,p)\leq(1+\sqrt{2})^{2}|sp|$. We prove the above by induction on the rank of the edge $\{s,p\}$, i.e., the place of the edge $\{s,p\}$ in a nondecreasing length order of the edges in $\DT(P)$.

\noindent
\textbf{Base case:} Let $\{s,p\}$ be the shortest edge in $\DT(P)$. Then, edge $\{s,p\}$ has been added to $E$ during the first iteration of the loop in step~\ref{loop1}, and therefore $\delta_G(s,p)=|sp|$.

\noindent
\textbf{Induction hypothesis:} Assume for every edge $\{r,q\}\in \DT(P)$ 
shorter than $\{s,p\}$, the Lemma holds, i.e., $\delta_G(r,q)\leq (1+\sqrt{2})^{2}|rq|$.

\noindent
\textbf{The inductive step:} 
If $\{s,p\} \in E$, we are done. 
Otherwise, w.l.o.g. assume $\{s,p\} \in C_{s_i}$ and $\{s,p\} \in C_{p_j}$; then, there exists either an edge $\{s,r\}\in C_{s_i} \cap E$, such that $|sr|\leq|sp|$, or an edge $\{p,r\}\in C_{p_j} \cap E$, such that $|pr|\leq|sp|$. Assume w.l.o.g. there exists an edge $\{s,r\}\in C_{s_i} \cap E $, such that $|sr|\leq|sp|$. By Lemma~\ref{lemma:shortest}, for every $x \in S_{s,r,p}$, $|sx| \geq min\{|sr|,|sp|\}=|sr|$. Let $\{r,t\}$ be the first  edge in $P_{S_{s,p,r}}$, and $\{q,p\}$ the last. Note that all edges of $P_{S_{s,t,q}}$ have been added to $E$ during Algorithm~\ref{alg:wedge}. 
\begin{claim} 
     The edges $\{r,t\}$ and $\{q,p\}$ are shorter than $\{r,p\}$.
\end{claim}
\begin{proof}
If $e \in \{\{r,t\}, \{q,p\}\}$ is outside the triangle $\triangle(srp)$, by Observation~\ref{obs:pi-x} the angles $\angle(rtp), \angle(rqp) \geq \pi-\angle(rsp)\geq \frac{3\pi}{4}$, and therefore, are shorter than $\{r,p\}$. Otherwise, it is bounded inside the triangle $\triangle(prr')$, where $r'$ is a point on $\{s,p\}$, such that $|sr'|=|sr|$ (since $|sr| \leq |sp|$, such a point exists). Therefore, this edge is shorter than $\{r,p\}$. 
\end{proof}

Applying the induction hypothesis on $\{r,t\}$ and $\{q,p\}$ results in:
\begin{eqnarray}
\delta_G(r,t) & \leq & (1+\sqrt{2})^{2}|rt| \\
\delta_G(p,q) & \leq & (1+\sqrt{2})^{2}|pq|. 
\end{eqnarray}
 
Let $\{s,a\}$ be the shortest edge in the wedge $W_{s,t,q}$ (i.e., the closest point to $s$ in  $S_{s,r,p} \backslash \{r,p \}$). By Corollary~\ref{cor:pathBound}, 
\begin{eqnarray}
 \delta_{S_{s,t,a}} & \leq & \frac{\pi}{2}|ta| \\
\delta_{S_{s,a,q}}  & \leq  & \frac{\pi}{2}|aq|. 
\end{eqnarray}  
Since $|sr| \leq |st|$ and $|sa| \leq |st|$, the angle $\angle(rta)$ facing towards $s$ is less than $\pi$. By Observation~\ref{obs:pi-x} we get that $\angle(rta)\geq \frac{3\pi}{4}$. Applying Claim~\ref{claim:triangle} on the triangle $\triangle(rta)$, with $d = \frac{\pi}{2}$ gives us 
\begin{eqnarray}
(1+\sqrt{2})^{2}|rt|+ \frac{\pi}{2}|ta|  & \leq  & (1+\sqrt{2})^{2}|ra|.
\end{eqnarray}

Therefore, 
\begin{eqnarray*}
 \delta_G(s,p) & \leq^{\qquad} & |sr|+ \delta_G(r,p) \\
          & \leq^{\qquad} & |sr|+ \delta_G(r,t)+ \delta_G(t,q) +\delta_G(q,p) \\
         & \leq^{(1),(2)} & |sr|+(1+\sqrt{2})^{2}(|rt|+|pq|) + \delta_G(t,q) \\
        & \leq^{\qquad} & |sr|+(1+\sqrt{2})^{2}(|rt|+|pq|) + \delta_{S_{s,t,a}}+ \delta_{S_{s,a,q}} \\
        & \leq^{(3),(4)} & |sr|+(1+\sqrt{2})^{2}(|rt|+|pq|)  + \frac{\pi}{2}(|ta|+|qa|)\\
        & \leq^{(5)} & |sr|+(1+\sqrt{2})^{2}(|ra|+|pq|)  + \frac{\pi}{2}(|aq|).
\end{eqnarray*}  

There are two cases regarding the location of points $q$ and $a$:
\begin{itemize}
	\item \textbf{Case 1:} Either point $q$ or $a$ is inside the triangle $\triangle(srp)$.\\
 	Let $r'$ be a point on $\{s,p\}$ such that $|sr|=|sr'|$. Notice $|sr| \leq |sp|$, and therefore, such a point exists. Since $|sr| \leq |sa|$ and $|sr|\leq |sq|$, $q$ and $t$ lie outside the disk centered at $s$ and with radius $|sr|$. Therefore, either point $q$ or point $a$ is located inside the triangle $\triangle(rr'p)$.
 
   Since $(1+\sqrt{2})^{2} > \frac{1}{1-2 \sin(\pi/8)}$, by Claim~\ref{cl:rr'AndSp} we get,
   $$|sr| + (1+\sqrt{2})^{2}(|rr'|+|r'p|) \leq (1+\sqrt{2})^{2}|sp|.$$
   Therefore, it is enough to show that 
 	$$(1+\sqrt{2})^{2}(|ra|+|pq|)  + \frac{\pi}{2}(|aq|) \leq (1+\sqrt{2})^{2}(|rr'|+|r'p|).$$
 	
 	Observe the following two cases regarding the convexity of the polygon $(raqp)$:
 	\begin{itemize}
 	\item \textbf{Case 1.1:} The polygon $(raqp)$ is convex.\\
 	Since $\frac{\pi}{2} < (1+\sqrt{2})^{2}$, we get
	\begin{eqnarray*}
	(1+\sqrt{2})^{2}(|ra|+|qp|) + \frac{\pi}{2}(|aq|) 
	                              & < & (1+\sqrt{2})^{2}(|ra|+|qp|+|aq|) \\
 	                              & \leq & (1+\sqrt{2})^{2}(|rr'|+|r'p|).
\end{eqnarray*}
 The last inequality follows from the convexity of the polygon $(raqp)$. 

 \item \textbf{Case 1.2:} The polygon $(raqp)$ is not convex.\\
 	The vertex that violates the convexity is either $a$ or $q$, and the other vertex is inside triangle  $\triangle(srp)$ (see,~Figure~\ref{fig:s.f_case1}).	 Assume w.l.o.g. that $a$ is the vertex that violates the convexity; then the angle $\angle(raq)$ facing towards $s$ is less than $\pi$. By Observation~\ref{obs:pi-x} $\angle(raq)\geq \frac{3\pi}{4}$; therefore, applying Claim~\ref{claim:triangle} on the triangle $\triangle(raq)$ with $d = \frac{\pi}{2}$ gives us
\begin{eqnarray}
\label{eq:eq5}  	
 	 (1+\sqrt{2})^{2}|ra| +  \frac{\pi}{2}|aq| \leq  (1+\sqrt{2})^{2}|rq|.
 \end{eqnarray} 	
 	
 	Thus, we get
\begin{eqnarray*} 	 
 	(1+\sqrt{2})^{2}(|ra|+|qp|)+ \frac{\pi}{2}(|aq|) & \leq^{(\ref{eq:eq5} )} & (1+\sqrt{2})^{2}(|rq|+|qp|) \\
 	& \leq^{(*)} &  (1+\sqrt{2})^{2}(|rr'|+|r'p|)
 \end{eqnarray*}
 The last inequality (*) follows from the convexity of triangle $\triangle(rqp)$.  	%
\end{itemize}

\begin{figure}[htp]
    \centering
        \includegraphics[width=1\textwidth]{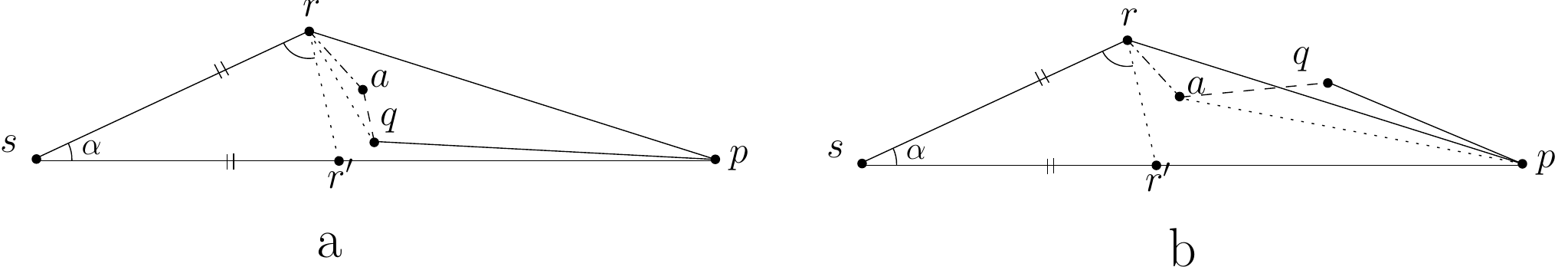}
            \caption{Illustrating the proof of Lemma~\ref{lemma:lemma5}, Case 1.2. In Figure (a) point $a$ violates the convexity of polygon $(raqp)$, while in Figure (b) point $q$ violates the convexity of the polygon.}
    \label{fig:s.f_case1}
\end{figure}

\item \textbf{Case 2:} Both points $q$ and $a$ are outside the triangle $\triangle(srp)$.\\
	 In this case, the edges added to $E$ depend on the angle $\angle(srt)$. There are two cases:
 \begin{itemize}
 		\item \textbf{Case 2.1:} Angle $\angle(srp)\geq \frac{\pi}{2}$. \\
 		In this case, $\angle(srt)\geq \angle(srp)\geq \frac{\pi}{2}$, and therefore, the algorithm	adds the edge $\{r,t\}$ to $E$. Thus, instead of showing 	
 		$$|sr|+(1+\sqrt{2})^{2}(|ra|+|qp|)+ \delta_{S_{s,a,q}}\leq (1+\sqrt{2})^{2}|sp|,$$
 		it is enough to show $$|sr|+(1+\sqrt{2})^{2}|qp|+ \delta_{S_{s,r,q}}\leq (1+\sqrt{2})^{2}|sp|.$$	By Corollary~\ref{cor:pathBound}, 
 		\begin{eqnarray}
 		\label{eq:eq7}   
 		   \delta_{S_{s,r,q}}\leq \frac{\pi}{2}|rq|.
 		\end{eqnarray}   
 		By Observation~\ref{obs:pi-x}, $\angle(rqp)\geq \frac{3\pi}{4}$, and by applying Claim~\ref{claim:triangle} on the triangle $\triangle(rqp)$ with $d = \frac{\pi}{2}$ we get
 \begin{eqnarray}
 		\label{eq:eq8}
 		(1+\sqrt{2})^{2}|qp| + \frac{\pi}{2}|rq| \leq (1+\sqrt{2})^{2}|rp|.
 		\end{eqnarray}
 		Thus, 
 		\begin{eqnarray*}
 		 |sr|+(1+\sqrt{2})^{2}|qp|+ \delta_{S_{s,r,q}} 
 		     &\leq^{(\ref{eq:eq7})} & |sr|+(1+\sqrt{2})^{2}|qp| + \frac{\pi}{2}|rq| \\
 		     &\leq^{(\ref{eq:eq8})} & |sr|+(1+\sqrt{2})^{2}|rp| \\
 	       &\leq^{(**)} & |sr|+(1+\sqrt{2})^{2}(|rr'| +|pr'|) \\
 	       &\leq^{(***)} &  (1+\sqrt{2})^{2} |sp|
 		 \end{eqnarray*}
 		Inequality (**) follows from triangle inequality for any point $r'$, and thus it also holds for a point $r'$ on $\{s,p \}$, such that $|sr|=|sr'|$. Therefore, from Claim~\ref{cl:rr'AndSp} and since $(1+\sqrt{2})^{2}  > \frac{1}{1-2 \sin (\pi/8)}$, inequality (***) follows.  
		
	\item \textbf{Case 2.2:} Angle $\angle(srp)< \frac{\pi}{2}$. \\
      Since $a$ and $q$ are outside $\triangle(rps)$, either angle $\angle(raq)$ or angle $\angle(aqp)$ (facing towards $s$) is less than $\pi$. Assume w.l.o.g. that $\angle(raq) < \pi$, thus, by Observation~\ref{obs:pi-x}, $\angle(raq) \geq \frac{3\pi}{4}$. Applying Claim~\ref{claim:triangle} on triangle $\triangle(raq)$ with $d = \frac{\pi}{2}$ gives us
 		\begin{eqnarray}
 		\label{eq:eqTr2}
 		   (1+\sqrt{2})^{2}|ra| + \frac{\pi}{2}|aq| \leq (1+\sqrt{2})^{2}|rq|.
     \end{eqnarray}
     Let $q'$ be a point on the intersection of disk $D_{s,r,p}$ and the extension of $\{s,q\}$ (see,~Figure~\ref{fig:sf_case22}). Then, by convexity we get 
     \begin{eqnarray}
 		\label{eq:eqTrIn}
 		   |rq| + |pq| \leq |rq'| + |pq'| \leq^{(*)}  \frac{|pr|}{ \cos(\alpha / 2)}.
 		 \end{eqnarray}  
      The last inequality (*) is obtained by Claim~\ref{cl:maxEq}.

      Let $b$ be a point on the extension of $\{s,r \}$, such that $|sb|= |sp|$; therefore, 
      $\angle(sbp) = \angle(spb) = \frac{\pi}{2}-\frac{\alpha}{2}$.\\
      By the law of sines,
      \begin{eqnarray}
      \label{eq:pbsp}
           |pb|  =  |sp|\frac{\sin(\alpha)}{\sin(\frac{\pi}{2}-\frac{\alpha}{2})}
            		 =  |sp|\frac{\sin(\alpha)}{\cos(\frac{\alpha}{2})}
            		 =  2|sp|\sin(\alpha/2).
     \end{eqnarray}
      
      Since  angle $\angle(srp) < \pi/2$, it follows that angle  $\angle(brp) > \pi/2$, thus,  
      \begin{eqnarray}
      \label{eq:pbpr}
         |pb| > |pr|.
      \end{eqnarray}
       
       Now we are ready to bound the length of the path: 
      \begin{eqnarray*}
      \label{eq:fin}
        |sr|+ \delta_G(r,p) 
          &\leq^{\ \ \ } &  
                 |sr|+(1+\sqrt{2})^{2}(|ra|+|pq|)  + \frac{\pi}{2}|qa|   \\
          & \leq^{ (\ref{eq:eqTr2}) } &  
                 |sr|+(1+\sqrt{2})^{2}(|rq|+|pq|)   \\
          & \leq^{ (\ref{eq:eqTrIn}) } &
                |sr|+(1+\sqrt{2})^{2}\frac{|pr|}{\cos(\alpha/2)}     \\
          & \leq^{ (\ref{eq:pbpr}) } &  
                |sr|+(1+\sqrt{2})^{2} \frac{|pb|}{\cos(\alpha/2)}   \\
         & \leq^{ (\ref{eq:pbsp}) } &  
                |sb|+(1+\sqrt{2})^{2}\frac{2|sp|\sin(\alpha/2)}{\cos(\alpha/2)} \\
         & =^{\ \ \ } &  
                 |sp|(1+(1+\sqrt{2})^{2}(2 \tan(\alpha/2) )) \\
         & \leq^{(**)}& 
                 |sp|(1+(1+\sqrt{2})^{2}(2 \tan(\pi/8) ))\\
         & =^{\ \ \ }& 
                 |sp|(1+2(1+\sqrt{2})(1+\sqrt{2})(\sqrt{2}-1))\\
         & =^{\ \ \ }& 
                 |sp|(1+2(1+\sqrt{2}))\\
         & =^{\ \ \ }& 
                 |sp|(1+\sqrt{2})^{2}.             
  \end{eqnarray*}
    The last inequality (**) follows from the fact that \emph{tangent} is a monotone increasing function in the range $(0,\pi/4]$. 
      
 	\begin{figure}[htp]
		    \centering
		        \includegraphics[width=0.25\textwidth]{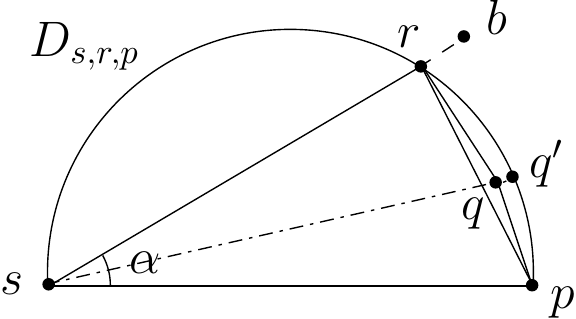}
		    \caption{Illustrating the proof of Lemma~\ref{lemma:lemma5}, case 2, when $\angle(srp)< \frac{\pi}{2}$.}
		    \label{fig:sf_case22}
		\end{figure}
		
		\end{itemize}
	\end{itemize}
	\end{proof}
 		
\vspace{0.2cm}
\begin{claim}\label{cl:maxEq}
Let $a$, $b$, and $c$ be three points on a circle, such that $\angle(abc)= \pi-\alpha$. Then, $|ab| + |bc| \leq \frac{|ac|}{\cos(\frac{\alpha}{2})}$.
\end {claim}

\begin{proof}
Let $\beta_1$ be the angle between $ba$ and $ca$, and let $\beta_2$ be the angle between $bc$ and $ca$, as depicted in Figure~\ref{fig:maxEq}. By the law of Sines we have $$ \frac{|ac|}{\sin(\pi-\alpha)}=\frac{|bc|}{\sin(\beta_1)}= \frac{|ab|}{\sin(\beta_2)}.$$  
Therefore,
$$|ab| + |bc| =  |ac|\frac{\sin(\beta_2)}{\sin(\pi-\alpha) } + |ac|\frac{\sin(\beta_1)}{\sin(\pi-\alpha)} = 
\frac{|ac|}{\sin(\pi-\alpha))}(\sin(\beta_2) + \sin(\beta_1)).$$ 
For $0 \leq \alpha \leq \pi/2$, this function is maximized when $\beta_1 = \beta_2 = \frac{\alpha}{2}$. Thus,
\begin{eqnarray*}
|ab| + |bc| &\leq& 2|ac|\frac{\sin(\frac{\alpha}{2})}{\sin(\pi-\alpha)} \\
&=& 2|ac|/\frac{\sin(\alpha)}{\sin(\frac{\alpha}{2})}\\
&=& \frac{2|ac|}{2\cos(\frac{\alpha}{2})}\\
&=& \frac{|ac|}{\cos(\frac{\alpha}{2})}.
\end{eqnarray*} 
  
\end{proof}
 		 
\begin{figure}[htp]
    \centering
        \includegraphics[width=0.5\textwidth]{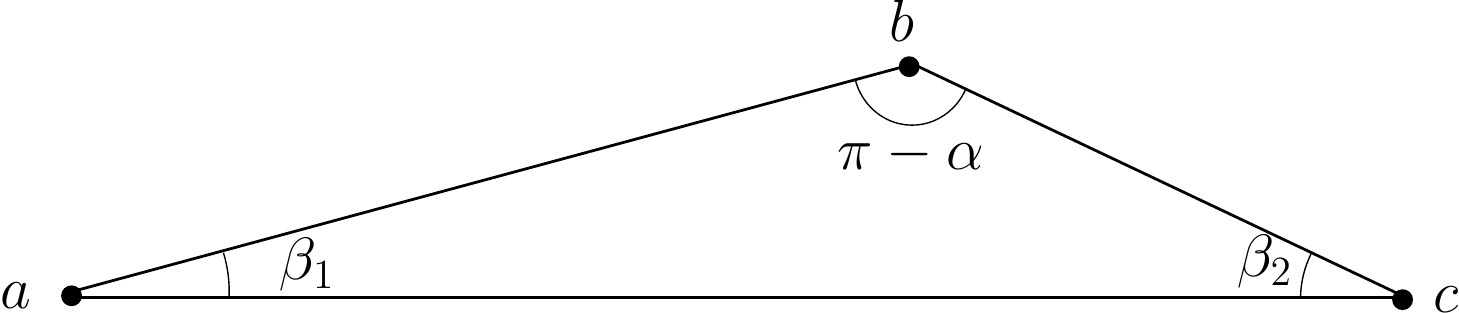}
    \caption{Illustrating the proof of Claim~\ref{cl:maxEq}.}
    \label{fig:maxEq}
\end{figure}

\begin{theorem}\label{theo:main}
For every set of points $P$, there is a strong planar $t$-spanner with 
$t = (1+\sqrt{2})^{2}\cdot\delta$, where $\delta$ is the stretch factor of Delaunay triangulation with bounded degree 7.
\end{theorem}

\section{Linear-running-time distributed algorithm for planner t-spanner with bounded degree 7 } \label{sec:runTime}
In this section we show a linear-running-time distributed algorithm that, given a Delaunay triangulation of set of points $P$ and clockwise order of the edges around each point $p \in P$, computes a planner $t$-spanner of $P$ with bounded degree 7, where the stretch factor $t$  is as before ($(1+\sqrt{2})^2 \delta$). 

The distributed algorithm chooses the edges according to the same principle as in Algorithm~\ref{alg:AlgMain}, meaning Lemma~\ref{lemma:shortest} still holds, only this time the edges that could be chosen by the algorithm are predefined. 
In order to predefine the candidate edges, we need to observe which kind of edges in a cone could be chosen by the algorithm.  Moreover, we need to compute and store these edges in a sorted order by length in linear time. 
The linear running time is achieved by selecting the edges based on several properties rather than sorting them. For simplicity of presentation we assumed that all edges in a cone are of different length, however even if this is not the case, we can solve it using known method in distribute computing such as \emph{wait and notify}.
Notice that Step~\ref{step:whileSy} in Algorithm~\ref{alg:AlgMainLinDis} terminates, since we are in the distributed setting. Moreover, each time-step there exist at least two points that remove at least one edge from the top of their lists. Since there are at most linear number of edges that can be removed, and each edge removal takes constant time, the total running time of the algorithm is $O(n)$.

\begin{algorithm}[htp]
\caption{BoundSpanner($P$)}\label{alg:AlgMainLin}
\begin{algorithmic}[1]
\REQUIRE $\DT(P)$ - Delaunay triangulation of set of points $P$ 
                   and a clockwise order of the edges around each point
\ENSURE A planar $t$-spanner $G = (P,E')$ with maximum degree 7 
\vspace{0.2cm}
\STATE  Initialize $C_p$ for each $p \in P$     \hspace{3.5cm} /* with respect to edge $(p,q_{min})$ */  \label{init1}

\FOR {every cone $C_{p_j}$ in $C_p$}
	\STATE CandidateEdgesInCone($C_{p_j}$)  \hspace{0.3cm} /* calling a subroutine that returns a set of edges \\
	                                        \hspace{5.8cm}    in a cone that are candidates to be added to the    \\
	                                        \hspace{5.8cm}    set $E'$ in a sorted order by length \  */ \label{init2}
\ENDFOR 
\STATE $E' \leftarrow \emptyset$
\STATE $E^* \leftarrow \emptyset$  \hspace{3.9cm}    /* additional edges to be included in $E'$ in the end */ 
\STATE Let $List(p)$ be the sorted list obtained by merging the lists $List(C_{p_j})$ for all $C_{p_j} \in C_p$ 
\vspace{0.1cm}
\FOR {each $p$ in $P$ \  (* in distributed behavior *) } 
	    \STATE DistributedEdgeSelction(p) 
\ENDFOR	     
\STATE $E' \leftarrow E'\cup E^*$     

\end{algorithmic}
\end{algorithm}

\begin{algorithm}[htp]
\caption{ CandidateEdgesInCone(cone c)}\label{alg:AlgCones}
\begin{algorithmic}[1]
\REQUIRE A cone $c$ in $C_p$ 
\ENSURE A sorted list $List(c)$ by length of edges in the cone $c$ that are candidates to be added to the set $E'$ 
\vspace{0.2cm}

 \STATE Let $\{\{p,p_1\}, \dots ,\{p,p_k\}\}$ be the edges in the cone $c$ in clockwise order
 \STATE Initialize $List_1(c) \leftarrow \{\{s,p_1\}\}$ and $List_2(c) \leftarrow \{\{s,p_k\}\}$
 \STATE $i \leftarrow 2$
 \WHILE {$|\{p,p_i\}| \geq |\{p,p_{i-1}\}|$}
 		\STATE Add $\{p,p_i\}$ to $List_1(c)$
 		\STATE $i \leftarrow i + 1$
 \ENDWHILE
 \STATE $i \leftarrow k-1$.
 \WHILE {$|\{p,p_i\}| \geq |\{p,p_{i+1}\}|$}
 		\STATE Add $\{s,p_i\}$ to $List_2(c)$
 			\STATE $i \leftarrow i - 1$
 \ENDWHILE 
 
 \STATE $List(c) \leftarrow$ merge($List_1(c)$, $List_2(c)$) \hspace{2.7cm}  /* merging the lists by length */
   \STATE Add the shortest edge in $c$ excluded edges in $List(c)$ in the appropriate place in the sorted list $List(c)$ 
                 \hspace{1.5cm}	/* in case there are more than one edge - add them all */
\end{algorithmic}
\end{algorithm}

\begin{algorithm}[htp]
\caption{DistributedEdgeSelction(p)}\label{alg:AlgMainLinDis}
\begin{algorithmic}[1]
\REQUIRE Sorted list $List(p)$ 
\ENSURE The spanner edges in $p$ contributed by $p$ 
\vspace{0.2cm}
	\WHILE {($List(p) \neq \emptyset$)}\label{st:headOfList} 
	   \STATE $\{p,q \} \leftarrow \ TOP(List(p))$ \qquad \qquad /* $TOP(List(p))$ is the first edge in $List(p)$  */   
	   \STATE Remove $\{p,q \}$ from $List(p)$
	   \IF {not $\forall C_{p_i}$ contain $\{p,q  \}$, $C_{p_i}\cap E' = \emptyset$} 
	       \STATE  Go back to~\ref{st:headOfList}
	   \ENDIF
	   \STATE \textbf{while}  ($|TOP(List(q))| < |\{p,q\}|$) \textbf{do} \ (nothing) \ /* distributed behavior ensures \\ 
	                                                     \hspace{8.7cm}                   while-loop termination \ */ \label{step:whileSy} 
	   \IF {($\forall C_{q_j}$ contain $\{p,q  \}$, $C_{q_j}\cap E' = \emptyset$)}   
				\STATE $E' \leftarrow E'\cup \{ \{p,q  \} \}$ 
        \STATE Wedge($p,q$) 
        \STATE Wedge($q,p$) 
		\ENDIF
	\ENDWHILE
\end{algorithmic}
\end{algorithm}

\begin{lemma}\label{lemma:5}
Let $\{s,p_1\},...,\{s,p_k\}$ be all the edges of a cone $c \in C_p$ in clockwise order. 
Let $W^*_{s,p_1,p_i}$ be the maximum wedge such that 
$|s p_{m-1}| \leq |s p_m|$ for every $1 < m \leq i$.
Symmetrically, let $W^*_{s,p_j,p_k}$ be the maximum wedge such that $|s p_{m+1}| \leq |s p_m| $ for every $j \leq m < k$. I.e., wedge $W*_{s,p_1,p_i}$  (alternatively,  $W*_{s,p_j,p_k}$) is the largest increasing sequence of edges clockwise (counterclockwise) starting from $\{s,p_1\}$ ($\{s,p_k\}$, respectively).
Let $E_{short}$ be the set of shortest edges in the wedge $W_{s,p_i,p_j}$, 
i.e., $\forall e_1 \in  E_{short}$ and $\forall e_2 \in W_{s,p_i,p_j}$ it holds that $|e_1| \leq |e_2|$. 
Then, the edge of cone $c$ that has been chosen by Algorithm~\ref{alg:AlgMain} to be added to $E'$, is from the set 
$E_{short} \cup W^*_{s,p_1,p_i} \cup W^*_{s,p_j,p_k}$.
\end{lemma}

\begin{proof}
Assume on the contrary that Algorithm~\ref{alg:AlgMainLin} chooses an edge $\{s,p_t\} \in c$ to be added to $E'$ that is not in $E_{short} \cup W^*_{s,p_1,p_i} \cup W^*_{s,p_j,p_k}$. 
W.l.o.g. (due to symmetry) assume $\{s,p_t\} \in W_{s,p_i,p_r}\backslash \{\{s,p_r\}\}$ for $\{s,p_r\} \in E_{short}$. Observe $\{s,p_j\}$, due to the maximality of $W^*_{s,p_j,p_k}$, we have $|s p_r| < |s p_j|$. 
Since $\{s,p_r\}$ is the shortest edge in the wedge $W_{s,p_i,p_j}$ and since $\{s,p_t\} \notin E_{short}$, it implies that $|s p_r| < |s p_t|$. Therefore, $|s p_r| < min\{|s p_t|,|s p_j|\}$ in contradiction to Lemma~\ref{lemma:shortest} (see,~Figure~\ref{fig:linear}).
\end{proof} 

\begin{figure}[htp]
    \centering
        \includegraphics[width=0.3\textwidth]{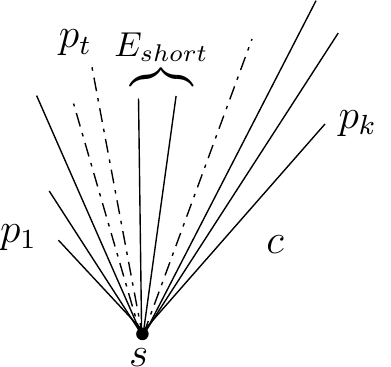}
    \caption{Illustrating the proof of Lemma~\ref{lemma:5}. The edges with no potential to be chosen by Algorithm~\ref{alg:AlgMain} are dashed. }
    \label{fig:linear}
\end{figure}

\begin{lemma}\label{lemma:sameEdges}
The resulting graph $G=(P,E)$ of Algorithm~\ref{alg:AlgMain} and the resulting graph $G'=(P,E')$ of Algorithm~\ref{alg:AlgMainLin} are identical.
\end{lemma}

\begin{proof}
Note that the initialization of the cones for each point is the same in both algorithms. Moreover, similarly to Algorithm~\ref{alg:AlgMain}, Algorithm~\ref{alg:AlgMainLin} chooses an edge from every cone to be added to $E'$. By Lemma~\ref{lemma:5}, for every cone Algorithm~\ref{alg:AlgMainLin} chooses an edge among all possible edges Algorithm~\ref{alg:AlgMain} could choose for that cone. 
Not only that, but the edges are chosen by the same principle -  the shortest edge, whose cones from both sides "agree" on, is added to $E$. Therefore, $E=E'$ and $G=G'$.
\end{proof}

\begin{theorem}
Algorithm~\ref{alg:AlgMainLin} has the following properties:
\begin{enumerate}
	\item The output graph of the algorithm has a bounded degree 7.
	\item The stretch factor of the output graph is $(1+\sqrt{2})^{2}\delta$, where $\delta$ is the stretch factor of Delaunay triangulation.
	\item The running time of the algorithm is linear.
\end{enumerate}
\end{theorem}

\begin{proof}

\begin{enumerate}	
	\item The degree bound follows exactly the degree bound of the previous algorithm (Algorithm~\ref{alg:AlgMain}) and 
	thus is at most 7.
	\item By Lemma~\ref{lemma:sameEdges}, the output graph of Algorithms~\ref{alg:AlgMain} and~\ref{alg:AlgMainLin} are identical. Therefore, the stretch factor of Algorithm~\ref{alg:AlgMainLin} is the same as the one of Algorithms~\ref{alg:AlgMain}, which is $(1+\sqrt{2})^{2}\delta$.
	
	\item Steps~\ref{init1} to~\ref{init2} in Algorithm~\ref{alg:AlgMainLin} require traversal over all points and edges of $\DT(P)$. Since $\DT(P)$ contains $O(n)$ edges we get $O(n)$ running time. 
				Merging two sorted lists and inserting edges of the same length cost $O(n)$ time. 
				The distributed subroutine requires every point to remove all edges incident to it from their list 
				and wait, in the worst case, until all edges have been removed from their lists. Since every edge is incident to two points, the overall running time of this subroutine for each point is as the number of edges which is $O(n)$. 
				Since all points perform this subroutine in a distributed manner, we get $O(n)$ running time.
\end{enumerate}
\end{proof}

\bibliographystyle{abbrv}

\setlength{\baselineskip}{0.97 \baselineskip} 
\bibliography{ref}  

\end{document}